\documentclass[12pt]{article}
\usepackage[margin=1.3in]{geometry}
\usepackage[centertags]{amsmath}
\usepackage{amsfonts,amsthm,amssymb}
\usepackage{amssymb}
\usepackage{amsmath}
\usepackage{graphicx}
\usepackage{bbm}
\usepackage{color}
\usepackage{mathtools}
\usepackage{multirow}
\usepackage{url}

\def\E{\mathbb{E}}
\def\P{\mathbb{P}}

\def\y{\boldsymbol{YES}}
\def\n{\boldsymbol{NO}}

\theoremstyle{definition}
\newtheorem{theorem}{Theorem}
\newtheorem{corollary}{Corollary}
\newtheorem{definition}{Definition}
\newtheorem{lemma}{Lemma}

\newtheorem{example}{Example}
\newtheorem{question}{Question}

\begin{document}

\title{Empirical Distribution of Equilibrium Play and \\ Its Testing Application} 

\author{Yakov Babichenko\footnote{Center for the Mathematics of Information, California Institute of Technology. E-mail: yakovbab@tx.technion.ac.il.} \and Siddharth Barman\footnote{Center for the Mathematics of Information, California Institute of Technology. E-mail: barman@caltech.edu.} \and Ron Peretz\footnote{Department of Mathematics, London School of Economics. E-mail: ronprtz@gmail.com}}

\date{}
\maketitle

\begin{abstract}

We show that in any $n$-player $m$-action normal-form game, we can obtain an approximate equilibrium by sampling any mixed-action equilibrium a small number of times. We study three types of equilibria: Nash, correlated and coarse correlated. For each one of them we obtain upper and lower bounds on  the number of samples required for the empirical distribution over the sampled action profiles to form an approximate equilibrium with probability close to one.

These bounds imply that using a small number of samples we can test whether or not players are playing according to an approximate equilibrium, even in games where $n$ and $m$ are large. In addition, our results substantially improve previously known upper bounds on the support size of approximate equilibria in games with many players. In particular, for all the three types of equilibria we show the existence of approximate equilibrium with support size polylogarithmic in $n$ and $m$, whereas the previously best-known upper bounds were polynomial in
$n$~\cite{HRS,GL,HMbook}.

\end{abstract}

\section{Introduction}


Equilibria are central solution concepts in game theory. They are widely used as models to explain the observed behavior of self-interested entities like human players in strategic settings. Arguably the most prominent examples of such notions of rationality are Nash equilibrium~\cite{Nash}, correlated equilibrium~\cite{aumann74}, and coarse correlated equilibrium~\cite{hannan}.  At a high level, these concepts denote distributions over players' action profiles where no player can benefit, in expectation, by unilateral deviation. Though, in most empirical applications of game theory these underlying distributions are not explicitly known to an observer.  Rather, what one observes is the behavior of the players, i.e., players' realized actions. 

%

This naturally leads us to consider a setting in which players implement an underlying distribution---i.e., a mixed strategy---during multiple plays of the same game. Here, the mixed strategy is not known to an outside observer. Rather, the observer sees the (pure) actions selected by the players during the play; in other words, she observes independent and identically distributed (i.i.d.) samples from the mixed strategy. This framework captures typical (empirical) applications of game theory, and entails the following fundamental question: How many samples from an equilibrium (Nash, correlated, or coarse correlated) play are required to ensure that the induced empirical distribution forms an approximate equilibrium\footnote{An approximate equilibrium, with approximation parameter $\varepsilon >0$, is a distribution over action profiles at which no player has more than an $\varepsilon$ incentive to deviate.} (again, Nash, correlated, or coarse correlated)?  The main objective of this paper is to show that even in large games---i.e., games with a large number of players and/or actions---an extremely small number of samples generate an approximate equilibrium with high probability.  This result has several useful interpretations. 

1. \textbf{Testing whether players are playing according to an equilibrium.} In many strategic settings, it is important to test whether players are playing according to an equilibrium or not, but experimental data is limited and costly. Such scenarios are studied throughout experimental economics. In such contexts it is desirable to have tests that are reliable and require a small number of data points. Another case wherein this testing exercise is relevant is when the same game is played multiple times in independent environments. We observe a limited number of outcomes/data and our goal is to analyze, through the data, whether the agents are implementing an equilibrium. Our results (Theorems \ref{theo:ne-test}, \ref{theo:ce-test} and Corollary \ref{cor:cce-test}) show that we can accomplish this testing task even with a small dataset (i.e., few samples) that consists of i.i.d.~action profiles drawn from the underlying mixed strategy. Moreover the results show that the test can be performed via a direct algorithm: we simply need to check whether or not the empirical distribution of the observed data is an approximate equilibrium.

2. \textbf{Existence of \emph{simple} approximate equilibrium.} The existence of approximate equilibrium with support size polynomial in the number of players has been established in prior work (see~\cite{A}, \cite{LMM}, and~\cite{HRS}). In particular, these results show that in every normal-form game there exists an approximate equilibrium that is simply the uniform distribution over a small set of action profiles; such approximate equilibria are referred to as simple approximate equilibria. Our result extends this line of work by substantially improving the upper bounds on the support size of simple approximate equilibrium. In particular, Corollaries~\ref{cor:ne-ss} and~\ref{cor:cce-ss} along with Theorem \ref{theo:ce-ss} show that in every game there exists an approximate equilibrium of support size polylogarithmic in $n$ (the number of players) and $m$ (the number of actions per player), see Table \ref{tbl:2}.


3. \textbf{Population games.} Sampling occurs in life very naturally. In biology  \cite{SP}, for example, members of a species come in different traits: sex, size, color, etc. Every newborn has a trait which is sampled according to some distribution prescribed by its species. An ecological system is abstractly modeled as a game, where the players are species, the strategies are their traits, and mixed strategies represent distributions of traits among individuals of a species. This abstraction relies on the assumption that the size of the population is large enough to represent the mixed strategy of the species faithfully. The present paper provides estimates on what should be considered as ``large enough'' for such strategic reasoning.
The above situation is modeled in detail by a \emph{population game}. Here, each individual is modeled as a player, where members of the same species are identical players. It is assumed that individuals can't change their trait (sex, color, etc.); therefore it is reasonable to look at (approximate) pure equilibria. The diversity of traits is explained by assuming some source of randomization in the creation of individuals; therefore the special interest in sampled equilibria.  

Another justification for seeking such sampling-based pure equilibria was given by Kalai \cite{kalai04}, who viewed them as \emph{ex-post} equilibria. That is, the players play a mixed strategy equilibrium, and even after their strategies are realized they have no incentive to modify them. 

Our results can be interpreted as bounds on the minimal size of the population that ensures an ex-post equilibrium in a population game. Kalai \cite{kalai04} worked in the more general framework of semi-anonymous games. The minimal number of players needed for such an equilibrium to emerge was studied by Azrieli and Shmaya in the even more general framework of Lipschitz games \cite{AS}. It should be noted that the results we obtain for the special case of population games are stronger than what one could hope to obtain for arbitrary Lipschitz games.

4. \textbf{Short-Run Stability of Equilibrium in Repeated Games with Bounded Rationality}. The premise that players do not know their opponents' utility functions is a central construct in the study of \emph{uncoupled learning} in repeated games (see \cite{HMbook} and \cite{Ybook}). A reasonable assumption along these lines is that players do not know the mixed strategy of their opponents. This gives us a repeated-game model in which every player learns her opponents' mixed strategies by observing the actions that they play in each repetition of the game.\footnote{Note that this assumption holds in most repeated-game models, i.e., perfect monitoring.} In this framework, a natural way of learning opponents' mixed strategy is through the {empirical distribution} (i.e., approximating opponents' mixed strategy by the empirical distribution of the observed action profiles). Such a learning process has been  considered in the classical fictitious play literature~\cite{R}, and in many other recent results; e.g., the regret-testing dynamics (see \cite{FY} and \cite{GL2}).

Overall, we get the following key question in settings where players learn their opponents' mixed strategy through pure-action samples played by the opponents: How fast does the empirical distribution of equilibrium play forms an approximate equilibrium? 

The answer to this question sheds light on the short-run stability of equilibria. If the empirical distribution does not form an approximate equilibrium after multiple iterations (i.e., after multiple plays), then it is likely that an impatient player---who is uncertain about her opponents' strategies---will infer from the observed samples that her current mixed strategy is not a best reply to her opponents' strategies and deviate to some other strategy. This exact setup has been considered in \cite{FY} and \cite{GL2}. Our results imply that, even in large games, after a small number of iterations with high probability such a situation will not occur. I.e.,  the equilibria are stable in the short run. 

Specifically, we provide almost tight bounds for the above question (see Table \ref{table1}). These bounds in particular show that after a small number of iterations, each player can learn whether her current strategy is approximately a best reply to her opponents' strategies.

\subsection{Informal Statement of the Results}
We consider large normal-form games with $n$ players and $m$ actions per player (here at least one of the numbers $n$ or $m$ is large). Let $x$ be an equilibrium of the game (Nash, correlated, or coarse correlated), which is a distribution over the action profiles. We observe $k$ i.i.d.~samples from $x$.

For the case of Nash equilibrium, since the players are playing according to a product distribution, the correct notion for the \emph{empirical distribution of play} is the product of the empirical distributions of the actions played by each player. For the cases of correlated  and coarse correlated equilibrium the correct notion for the \emph{empirical distribution of play} is simply the empirical distribution of the sampled action profiles.

The empirical distribution of play is a random variable. In this paper we address the following question: how large should the number of samples $k$ be to guarantee that the empirical distribution of play forms an approximate equilibrium with probability close to $1$? We provide almost tight bounds for this question, which are summarized in the table below.

\begin{table}[h]
\begin{center}
\begin{tabular}{|ccc}
\hline 
\multicolumn{1}{|c|}{Equilibrium}      & \multicolumn{1}{|c|}{Upper Bound}          & \multicolumn{1}{|c|}{Lower bound}      \\ \hline
\multicolumn{1}{|c|}{Nash}              & \multicolumn{1}{|c|}{$k = O(\log m + \log n)$}    & \multicolumn{1}{|c|}{$k = \Omega (\log m + \log n)$}   \\
\multicolumn{1}{|c|}{}       & \multicolumn{1}{|c|}{Theorem \ref{theo:ne}}          & \multicolumn{1}{|c|}{Examples \ref{ex:al} and \ref{ex:n}} \\ \hline
\multicolumn{1}{|c|}{Correlated}        & \multicolumn{1}{|c|}{$k = O(m \log m + \log n)$} & \multicolumn{1}{|c|}{$k = \Omega (m + \log n)$}        \\
\multicolumn{1}{|c|}{}       & \multicolumn{1}{|c|}{Theorem \ref{theo:ce}}          & \multicolumn{1}{|c|}{Examples \ref{ex:cor-m} and \ref{ex:n}} \\ \hline
\multicolumn{1}{|c|}{Coarse Correlated} & \multicolumn{1}{|c|}{$k = O(\log m + \log n)$}    & \multicolumn{1}{|c|}{$k = \Omega (\log m + \log n)$}   \\
\multicolumn{1}{|c|}{}       & \multicolumn{1}{|c|}{Theorem \ref{theo:cce}}         & \multicolumn{1}{|c|}{Examples \ref{ex:al} and \ref{ex:n}} \\ \hline

\end{tabular}
\caption{Bounds on the number of samples that are required for the empirical distribution of a play to form an approximate equilibrium with probability close to 1.}\label{table1}
\end{center}
\end{table}


Moreover, we show that if players are playing according to a distribution that is \emph{not} an approximate equilibrium, then for the same values of $k$ (as in the upper bound column in Table \ref{table1}), with probability close to $1$, the empirical distribution of play will \emph{not} form an approximate equilibrium. Therefore, we can test whether players are playing according to an approximate equilibrium using $k$ samples (see Theorems \ref{theo:ne-test} and Corollaries \ref{theo:ce-test} and \ref{cor:cce-test}). These results suggests that even in games with a very large number of players or a very large number of actions there exists efficient tests to determine whether players are playing according to a Nash equilibrium or a coarse correlated equilibrium. Correlated equilibrium on the other hand is a slightly more complicated notion in this respect. We establish this by proving that there does exists a test for approximate correlated equilibrium that uses less than $\Omega(\sqrt{m})$ samples (see Theorem \ref{theo:cor-im}). This is in contrast to Nash equilibrium and coarse correlated equilibrium that require only $O(\log m)$ samples. 

The fact that the empirical distribution of play forms an approximate equilibrium with merely {positive} probability is interesting in its own right. The fact that this event occurs with positive probability establishes, via the probabilistic method, the existence of an approximate equilibrium with small \emph{support size}. Note that ``support size'' has different meanings in the case of Nash equilibrium and the cases of correlated and coarse correlated equilibrium. For Nash equilibrium the support size is the maximum number of actions that are played with positive probability by any single player. For correlated and coarse correlated equilibrium, the support size is the number of \emph{action profiles} that have a non-zero probability of being played in the approximate equilibrium. 

Small-support approximate Nash equilibria have been previously studied  in \cite{A}, \cite{LMM}, and \cite{HRS}. Alth\"{o}fer \cite{A} along with Lipton and Young~\cite{LY} studied the problem for two-player $m$-action zero-sum games and established a $\Theta(\log m)$ bound on the support size. Lipton, Markakis, and Mehta~\cite{LMM} studied the same question of general $n$-player $m$-actions games and achieved an upper bound of $O(n^2 \log m)$. H\'{e}mon et al. \cite{HRS} improved the bound of \cite{LMM} to $O(n\log m)$. Our result implies the existence of an approximate Nash equilibrium with support size $O(\log m + \log n)$ (see Corollary \ref{cor:ne-ss}). This result gives us an algorithm for computing an approximate Nash equilibrium with running time $O(N^{\log \log N})$ in games where  $m=poly (n)$ (i.e., the number of actions is not significantly larger than the number of players). Here $N$ is the input size of the game; see Corollary \ref{cor:ne-alg}. Note that the running time of the best previously known algorithm for this problem is $O(N^{\log N})$ (see \cite{N}). Hence, the current paper provides an exponential improvement upon prior work in the computation of approximate Nash equilibrium in large games. 

Support-size upper bounds for \emph{exact} correlated equilibria were studied by Germano and Lugosi in \cite{GL}. Specifically, they showed that every $n$-player $m$-action game admits a correlated equilibrium with support size $O(nm^2)$. Applying the technique of \cite{GL} we can obtain a bound of $O(nm)$ on the support size of exact coarse correlated equilibrium. The results in this paper prove that for \emph{approximate} correlated equilibrium and \emph{approximate} coarse correlated equilibrium the size of the support can be significantly reduced to $O(\log m (\log m + \log n))$ and $O(\log m + \log n)$ respectively.

Our bounds on the support size of approximate equilibria and previously known results are summarized in Table~\ref{tbl:2}.

\begin{table}[h]
\begin{center}
\begin{tabular}{ccc}
\hline
\multicolumn{1}{|c|}{Approximate Equilibrium} & \multicolumn{1}{|c|}{Our Results} & \multicolumn{1}{|c|}{Previous Bounds}                        \\ \hline
\multicolumn{1}{|c|}{Nash}              & \multicolumn{1}{|c|}{$O(\log m + \log n)$}          & \multicolumn{1}{|c|}{$O(n\log m)$}     \\
\multicolumn{1}{|c|}{}       & \multicolumn{1}{|c|}{Corollary \ref{cor:ne-ss}}     & \multicolumn{1}{|c|}{\cite{HRS}}       \\ \hline
\multicolumn{1}{|c|}{Correlated}        & \multicolumn{1}{|c|}{$O(\log m (\log m + \log n))$} & \multicolumn{1}{|c|}{Exact: $O(nm^2)$} \\
\multicolumn{1}{|c|}{}       & \multicolumn{1}{|c|}{Theorem \ref{theo:ce-ss}}      & \multicolumn{1}{|c|}{\cite{GL}}        \\ \hline
\multicolumn{1}{|c|}{Coarse Correlated} & \multicolumn{1}{|c|}{$O(\log m + \log n)$}          & \multicolumn{1}{|c|}{Exact: $O(nm)$}   \\
\multicolumn{1}{|c|}{}       & \multicolumn{1}{|c|}{Corollary \ref{cor:cce-ss}}    & \multicolumn{1}{|c|}{\cite{GL}}        \\ \hline
\end{tabular}
\caption{Support size of approximate equilibria.}\label{tbl:2}
\end{center}
\end{table}

\section{Notation and Preliminaries}

We consider $n$-player $m$-action games, i.e., games with $n$-player and $m$ actions per player.\footnote{All the results in the paper hold also for the case where each player has a different number of actions (i.e., player $i$ has $m_i$ actions). For ease of presentation, we assume throughout that all players have the same number of actions $m$.} Write $N:=nm^n$ to denote the input size of such games; i.e., the size of a list that enumerates players' utilities at every pure action profile. We use the following standard notation. The set of players is $[n]=\{1,2,...,n\}$. The set of actions of each player is $A_i=[m]=\{1,2,...,m\}$. The set of strategy profiles is $A=[m]^n$. The set of all probability distributions over a set $B$ is denoted by $\Delta(B)$. Therefore, $\Delta(A)$ is the set of all probability distributions over the action profiles, and $\Pi_{i\in [n]} \Delta(A_i)$ is the set of all \emph{product distributions}. For a vector $v=(v_j)_{j\in [n]}$, we denote by $v_{-i}:=(v_j)_{{j\neq i, j\in [n]}}$ the vector that does not contain the $i$'th coordinate. The payoffs of the players are normalized between $0$ and $1$. Specifically, the payoff function of player $i$ is denoted by  $u_i: A\rightarrow [0,1]$ and it can be extended to $u_i:\Delta(A)\rightarrow [0,1]$ by $u_i(x):=\E_{a\sim x} [u_i(a)]$.

\begin{definition}\label{def:k-uni}
A distribution over a set $B$ is called \emph{$k$-uniform} if it is the uniform distribution over a size-$k$ multiset of elements from $B$. Equivalently, $x\in \Delta(B)$ is $k$-uniform iff $x(b)=\frac{c_b}{k}$ for every $b\in B$ where $c_b \in \mathbb{N}$. The set of all $k$-uniform distributions on $B$ is denoted $\Delta_k(B)$.
\end{definition} 

%
%
%
%
%

\section{Approximate Nash Equilibrium}\label{sec:ne}

\begin{definition}
A product distribution $x=(x_i)_{i\in [n]}$ is an \emph{$\varepsilon$-Nash equilibrium} if no player can gain more than $\varepsilon$ by deviating to another strategy. Formally, $u_i(x)\geq u_i(a_i,x_{-i})-\varepsilon$ for every $i\in [n]$ and every $a_i\in A_i$.

When $\varepsilon =0$ we say that $x$ is an \emph{exact Nash equilibrium}, or simply \emph{Nash equilibrium}.
\end{definition}

Throughout the paper we refer to $\varepsilon$ as a \emph{constant}, and derive asymptotic results for games where at least one of the parameters $m$ or $n$ goes to infinity.

Assume that the players are playing according to a product distribution $x=(x_i)_{i\in [n]}$. We observe $k$ i.i.d.~\emph{samples} from $x$ that are denoted $(a(t))_{t\in [k]}$, where each $a(t)\in A$. Since we assume that the players are playing according to a product distribution, the correct interpretation of the observed data is as follows. We denote by $s^k_i$ the \emph{empirical distribution of player} $i$, defined to be the empirical distribution of the samples $(a_i(t))_{t\in [k]}$. Formally, $s^k_i(a_i) :=\frac{1}{k}|\{t :a_i(t)=a_i\}|$. The \emph{product empirical distribution of play} is the product distribution $\Pi_i  \  s^k_i$.

The following theorem states that if players are playing according to a Nash equilibrium then the product empirical distribution of play (which is a random variable) is an $\varepsilon$-Nash equilibrium after $k=O(\log n + \log m)$ samples, with probability close to $1$.

\begin{theorem}\label{theo:ne}
Let $x$ be a Nash equilibrium of an $n$-player $m$-action game and parameters $\varepsilon,\alpha \in (0,1)$. Then, the product empirical distribution of play $(s^k_i)_{i\in [n]}$ defined over $k$ i.i.d.~samples from $x$ is an $\varepsilon$-Nash equilibrium with probability greater than $1-\alpha$, for every 
\begin{align*}
k &> \frac{8(\ln m + \ln n - \ln \alpha - \ln \varepsilon + \ln 8)}{\varepsilon^2}=O(\log m + \log n + |\log \alpha|).
\end{align*}
\end{theorem}

We emphasize the \emph{logarithmic} dependence of the number of samples $k$ on the probability of error $\alpha$, which means that in order to reduce  the probability of error by a factor of two we should increase the number of samples only by a constant ($\frac{8\ln 2}{\varepsilon^2}$).

A proof of Theorem \ref{theo:ne} is given in Section \ref{sec:ne-proofs}. We note that the bound $O(\log m + \log n)$ is tight (up to a constant factor), see Examples~\ref{ex:n} and~\ref{ex:al} in Section \ref{sec:lb}.

\subsection{Existence of Simple Approximate Nash Equilibrium}

Note that if $s^k$ is an $\varepsilon$-Nash equilibrium with positive probability then \emph{there exists} a multiset of $k$ samples that forms an $\varepsilon$-Nash equilibrium; this claim follows from the probabilistic method. Note also that the empirical distribution of every player $i$ is a $k$-uniform distribution (see Definition \ref{def:k-uni}). This simple observation implies the following corollary from Theorem \ref{theo:ne}. Here, we say that an approximate Nash equilibrium $\prod_i s_i$ is $k$-uniform if every $s_i$ is a $k$-uniform distribution.  

\begin{corollary}\label{cor:ne-ss}
Every $n$-player $m$-action game admits an $k$-uniform $\varepsilon$-Nash equilibrium for every 
\begin{align*}
k> \frac{8(\ln m + \ln n - \ln \varepsilon + \ln 8)}{\varepsilon^2} =O(\log m + \log n).
\end{align*}
\end{corollary}

This corollary guarantees the existence of an approximate equilibrium (in every $n$-player $m$-action game) where each player uses only a small number of actions in her mixed strategy (at most $O(\log m + \log n)$). Another useful consequence, is the simplicity of the probabilistic structure of the mixed strategy of each player. To see it, consider, for example, the case of $n$-player $2$-action games. Here, corollary \ref{cor:ne-ss} with 
\begin{equation*}
k=\left\lfloor \frac{8}{\varepsilon^2}(\ln n + \ln \varepsilon + \ln 16) \right\rfloor+1=O(\log n),
\end{equation*}
implies that there exists an $\varepsilon$-Nash equilibrium in which each player $i$ uses a mixed strategy of the form $(\frac{c_i}{k},1-\frac{c_i}{k})$, where $c_i\in \mathbb{N}$.

The fact that such a simple approximate Nash equilibrium exists allows us to find an approximate Nash equilibrium just by exhaustively searching over all the possible $n$-tuples of $k$-uniform strategies. Although the algorithm is simple, to the best of our knowledge it provides the best-known running-time bound for this problem.

\begin{corollary}\label{cor:ne-alg}
Given an $n$-player $m$-action normal-form game and constant $\varepsilon>0$, there exists an algorithm that computes an $\varepsilon$-Nash equilibrium of the game in time $\min\{m^{nk},k^{nm}\}$, with $k=O(\log m + \log n)$.
\end{corollary}

Recall that $N=nm^n$ denotes the input size of $n$-player $m$-action games. For \emph{all} games we have 
\begin{equation*}
m^{nk}=poly(N^k) \leq poly(N^{n\log m})=poly(N^{\log N}),
\end{equation*}
which implies that the running time of the exhaustive search algorithm is at most $N^{O(\log N)}$. This bound coincides with the best-known upper bound for computing approximate Nash equilibrium (see \cite{N}). 

For the class of games where $m=poly(n)$ (e.g., $n$-player games with $n^2$-action games) the bound of Corollary \ref{cor:ne-alg} improves from $N^{\log N}$ to $N^{\log \log N}$:
\begin{equation*}
m^{nk}=poly(N^k) = poly(N^{\log n})=poly(N^{\log \log N}).
\end{equation*}
For the class of games where the number of actions is constant ($m=O(1)$), the bound of Corollary \ref{cor:ne-alg} further improves to $N^{\log \log \log N}$:
\begin{equation*}
k^{nm}=N^{\frac{nm\log k}{\log N}} = poly(N^{\log k})=poly(N^{\log \log \log N}).
\end{equation*}
\subsection{Testing Approximate Nash Equilibrium}\label{sec:ne-test}

Our goal is to test whether players are playing according to an approximate Nash equilibrium. The assumption is that the exact mixed strategies of the players are unobserved. Instead, we observe i.i.d.~samples from the underlying mixed strategies. Our focus is on the following question: how many samples are required to perform this test? 

Ideally, we would like to design a test that outputs the answer $\y$ (with probability close to $1$) if the players are playing according to $\delta$-Nash equilibrium, and it returns $\n$ (with probability close to $1$), otherwise. It is easy to see that such a test does not exist. The problem arises at the ``boundary''. Consider a distribution $x$ that is a $\delta$-Nash equilibrium, but every arbitrary small neighborhood of $x$ contains a distribution that is not a $\delta$-Nash equilibrium (it is easy to see that such a distribution $x$ always exists). Then the test should distinguish between $x$ and distributions that are arbitrary close to $x$ using a finite number of samples, which is impossible. Therefore, we weaken our requirements from a test by providing a slack of $\varepsilon$.

\begin{definition} 
Given a number of samples $k \in \mathbb{Z}_+$, a function $T:A^k\rightarrow \{\y, \n\}$ is said to be an \emph{$\varepsilon$-test that has error probability $\alpha$ for $\delta$-Nash equilibrium} if for every product distribution $x=(x_i)_{i\in [n]} \in \prod_i \Delta(A_i)$ we have
\begin{itemize}
\item $\P(T((a(t))_{t\in [k]})= \y)\geq 1-\alpha$, for every $x$ that is a $\delta$-Nash equilibrium;

\item $\P(T((a(t))_{t\in [k]})= \n)\geq 1-\alpha$, for every $x$ that is not a $(\delta + \varepsilon)$-Nash equilibrium.
\end{itemize} 
\end{definition}

In other words, we require that the test returns the correct answer, with probability close to $1$, for all distributions that are $\delta$-Nash equilibrium and for all distributions that are not $(\delta + \varepsilon)$-Nash equilibrium. We allow the test to return any answer when the distribution is a $(\delta + \varepsilon)$-Nash equilibrium but not a $\delta$-Nash equilibrium.

The following theorem states that using $O(\log n + \log m)$ samples we can test whether players are playing according to an approximate Nash equilibrium. Moreover, the test is very natural, we need to simply check whether the empirical distribution of play is an approximate Nash equilibrium or not.

\begin{theorem}\label{theo:ne-test}
Let 
\begin{align*}
T((a(t)_{t\in [k]})=\begin{cases} 
\y &\text{if } (s^k_i)_{i\in [n]} \text{ is a } (\delta + \frac{\varepsilon}{2}) \text{-Nash equilibrium}, \\
\n &\text{otherwise}. 
\end{cases}
\end{align*}
Then $T$ is an $\varepsilon$-test that has error probability $\alpha$ for $\delta$-Nash equilibrium, when the number of samples  
\begin{align*}
k & >\frac{72}{\varepsilon^2}(\ln (m+1) + \ln n -\ln \alpha - \ln \varepsilon + \ln 24)=O(\log m + \log n + |\log \alpha|).
\end{align*}
\end{theorem}

\hfill \\
Note that the number of samples is independent of $\delta$. Section \ref{sec:ne-proofs} contains a proof of Theorem~\ref{theo:ne-test}. 

\subsection{Proofs}\label{sec:ne-proofs}

The proofs of Theorems \ref{theo:ne} and \ref{theo:ne-test} are based on the following Lemma.

\begin{lemma}\label{pro:main}
For every $n$-player $m$-action game, every player $i\in [n]$, every action $a_i \in A_i = [m]$, and every product distribution of the opponents $x_{-i}=(x_j)_{j\neq i}$ we have
\begin{equation*}
\P(|u_i(a_i,s^k_{-i})-u_i(a_i,x_{-i})|\geq \varepsilon)\leq \frac{4e^{-\frac{\varepsilon^2}{2}k}}{\varepsilon}.
\end{equation*}
\end{lemma}

In other words, this lemma states that with probability that is exponentially (in $k$) close to 1, player $i$ is almost indifferent between the case where his opponents are playing the original distribution $x_{-i}$ or the product empirical distribution $s^k_{-i}$.

We emphasize that this lemma is a key technical contribution of this work that proves a novel concentration inequality for product distributions. Even though we apply the sampling method as in \cite{LMM} and \cite{HRS}, the fact that we use Lemma~\ref{pro:main} instead of some standard concentration inequality essentially enables us to significantly improve upon the previously-best-know bounds. A proof of Lemma~\ref{pro:main} is given in the appendix.

\begin{proof}[Proof of Theorem \ref{theo:ne}]
The proof uses similar idea to \cite{LMM} or \cite{HRS}. Lemma \ref{pro:main} and the choice of $k$ guarantee that

\begin{equation*}
\P(|u_i(a_i,s^k_{-i})-u_i(a_i,x_{-i})|\geq \frac{\varepsilon}{2})\leq \frac{8e^{-\frac{\varepsilon^2}{8
}k}}{\varepsilon}< \frac{\alpha}{mn},
\end{equation*}
for every player $i$ and every action $a_i\in [m]$. Using the union bound, we get that with probability greater then $1-\alpha$ we have $|u_i(a_i,s^k_{-i})-u_i(a_i,x_{-i})|< \frac{\varepsilon}{2}$, for \emph{all} players $i\in [n]$ and \emph{all} actions $a_i\in [m]$. In such a case $(s^k_i)_{i\in [n]}$ is an $\varepsilon$-Nash equilibrium because
\begin{multline*}
u_i(a_i,s^k_{-i}) \leq u_i(a_i,x_{-i})+\frac{\varepsilon}{2}\leq \sum_{a'_i\in A_i} s^k_i(a'_i) u_i(a'_i,x_{-i}) +\frac{\varepsilon}{2}\\
\leq \sum_{a'_i\in A_i} s^k_i(a'_i) u_i(a'_i,s^k_{-i}) +\varepsilon = u_i(s^k_i,s^k_{-i})+\varepsilon,
\end{multline*}
where the second inequality holds because all the strategies in the support of $s^k_i$ are in the support of $x_i$, which contains only best replies to $x_{-i}$.
\end{proof}

\begin{proof}[Proof of Theorem \ref{theo:ne-test}]
Lemma \ref{pro:main} and the choice of $k$ guarantee that
\begin{equation}\label{eq:in1}
\P(|u_i(a_i,s^k_{-i})-u_i(a_i,x_{-i})|\geq \frac{\varepsilon}{3})\leq \frac{24e^{-\frac{\varepsilon^2}{72
}k}}{\varepsilon}< \frac{\alpha}{(m+1)n},
\end{equation}
for every player $i$ and every action $a_i\in [m]$.
In addition, Hoeffding's inequality (see \cite{H}) guarantees that for a given $x_{-i}$ we have 
\begin{equation}\label{eq:in2}
\P(|u_i(s^k_i,x_{-i})-u_i(x_i,x_{-i})|\geq \frac{\varepsilon}{6})\leq 2e^{-\frac{\varepsilon^2}{72
}k}<\frac{24e^{-\frac{\varepsilon^2}{72
}k}}{\varepsilon}< \frac{\alpha}{(m+1)n},
\end{equation}  
for every player $i\in [n]$. Note that there are $n(m+1)$ inequalities of the form (\ref{eq:in1}) and (\ref{eq:in2}). Therefore, the union bound implies that with probability greater than $1-\alpha$ the following $n(m+1)$ inequalities hold:
\begin{align}\label{eq:in-all}
\begin{aligned}
|u_i(a_i,s^k_{-i})-u_i(a_i,x_{-i})| & \leq \frac{\varepsilon}{6} \qquad \forall i \in [n], \; \forall a_i \in [m]. \\
|u_i(s^k_i,x_{-i})-u_i(x_i,x_{-i})|& \leq  \frac{\varepsilon}{6} \qquad \forall i\in [n].
\end{aligned}
\end{align}
Throughout the proof we will assume that all the inequalities in (\ref{eq:in-all}) are satisfied.

If $(x_i)_{i\in [n]}$ is a $\delta$-Nash equilibrium then $(s^k_i)_{i\in [n]}$ is a $(\delta + \frac{\varepsilon}{2})$-Nash equilibrium because
\begin{multline*}
u_i(a_i,s^k_{-i}) \leq u_i(a_i,x_{-i})+\frac{\varepsilon}{6}\leq u_i(x_i,x_{-i}) +\delta+\frac{\varepsilon}{6}\leq u_i(s^k_i,x_{-i})+\delta + \frac{\varepsilon}{3} \\
= \sum_{a_i\in A_i} s^k_i(a_i) u_i(a_i,x_{-i}) +\delta + \frac{\varepsilon}{3} \leq \sum_{a_i\in A_i} s^k_i(a_i) u_i(a_i,s^k_{-i}) +\delta + \frac{\varepsilon}{2} =u_i(s^k_i,s^k_{-i})+\delta+ \frac{\varepsilon}{2}.
\end{multline*}

On the other hand, if $(x_i)_{i\in [n]}$ is not a $(\delta + \varepsilon)$-Nash equilibrium, then there exists a player $i$ and an action $a^*_i$ such that $u_i(a^*_i,x_{-i})>u_i(x_i,x_{-i})+\delta+\varepsilon$. In such a case $(s^k_i)_{i\in [n]}$ is not a $(\delta+\frac{\varepsilon}{2})$-Nash equilibrium because
\begin{multline*}
u_i(a^*_i,s^k_{-i}) \geq u_i(a^*_i,x_{-i})-\frac{\varepsilon}{6}> u_i(x_i,x_{-i}) +\delta+\frac{5\varepsilon}{6}\geq u_i(s^k_i,x_{-i})+\delta + \frac{4\varepsilon}{6} \\
= \sum_{a_i\in A_i} s^k_i(a_i) u_i(a_i,x_{-i}) +\delta + \frac{2\varepsilon}{3} \geq \sum_{a_i\in A_i} s^k_i(a_i) u_i(a_i,s^k_{-i}) +\delta + \frac{\varepsilon}{2} =u_i(s^k_i,s^k_{-i}) + \delta + \frac{\varepsilon}{2}.
\end{multline*}

Summarizing, the choice of $k$ guarantees that all the inequalities in (\ref{eq:in-all}) will be satisfied with probability of at least $1-\alpha$. If those inequalities are satisfied then we have the following:

- For every product distribution $x$ that is a $\delta$-Nash equilibrium, the product empirical distribution is a $(\delta+\frac{\varepsilon}{2})$-Nash equilibrium. Hence, for $\delta$-Nash equilibria, the given test  $T$ returns the correct answer $\y$.

- For every product distribution $x$ that is not a $(\delta+ \varepsilon)$-Nash equilibrium, the product empirical distribution is not a $(\delta+\frac{\varepsilon}{2})$-Nash equilibrium. Hence, for a distribution that is not  a $(\delta+ \varepsilon)$-Nash equilibirum, the given test $T$ returns the correct answer $\n$.
\end{proof}

\section{Approximate Correlated Equilibrium}\label{sec:ce}
Section \ref{sec:ne} considered the case of product distributions, i.e., a setting in which players followed their mixed strategies independently. We will now consider complementary notions of equilibria that address settings in which players' actions are correlated. Specifically, this section is focused on correlated equilibira and the next section is on coarse correlated equilibria.

A typical interpretation of correlated equilibrium is as follows. There exists a \emph{mediator} who samples an action profile $a=(a_i)_{i\in [n]}$ according to a distribution $x$. Then the mediator (privately) tells every player $i$ the corresponding action $a_i$. We will call the drawn action $a_i$ \emph{the recommendation to player $i$}. A distribution $x \in \Delta(A)$ is an $\varepsilon$-\emph{correlated equilibrium} if no player can gain more than $\varepsilon$ by deviating from the  recommendation of the mediator. A deviation from the mediator's recommendation is described by a \emph{switching rule} $f: A_i \rightarrow A_i$, that corresponds to the case where instead of the recommended action $a_i$ the player chooses to play $f(a_i)$.

\begin{definition}\label{defn:ce}
For every switching rule $f:A_i \rightarrow A_i$ we denote by $R^i_f(a):=u_i(f(a_i),a_{-i})-u_i(a_i,a_{-i})$ the regret of player $i$ for not implementing the switching rule $f$ at strategy profile $a$. 

A distribution $x\in \Delta(A)$ is an $\varepsilon$-\emph{correlated equilibrium} if $\E_{a\sim x} [ R^i_f(a) ] \leq \varepsilon$ for every player $i$ and every switching rule $f: A_i \rightarrow A_i$.
\end{definition}


Unlike the case of product distributions where it was reasonable to consider the \emph{product empirical distribution}, here in the case of general (not necessarily product) distributions we consider the empirical distribution of the sampled profiles. We assume that players are playing according to a distribution $x\in \Delta(A)$. We observe $k$ i.i.d.~\emph{samples} from $x$ that are denoted by $(a(t))_{t\in [k]}$ where $a(t)\in A$. Write $s^k$ for the \emph{empirical distribution of the samples}, specifically  $s^k(a):=\frac{1}{k}|\{t\in [k]: a(t)=a\}$.

We begin with stating the analogue of Theorem \ref{theo:ne} for the case of correlated equilibrium.

\begin{theorem}\label{theo:ce}
Let $x$ be a correlated equilibrium of an $n$-player $m$-action game and parameters $\varepsilon, \alpha \in (0,1)$. Then, the empirical distribution $s^k$ defined over $k$ i.i.d.~samples from $x$ is an $\varepsilon$-correlated equilibrium with probability greater than $1-\alpha$ for every
\begin{equation*}
k> \frac{2}{\varepsilon^2}(m\ln m + \ln n - \ln \alpha)=O(m\log m + \log n).
\end{equation*}
\end{theorem}

\hfill \\

The bounds of Theorem \ref{theo:ce} are almost tight. Specifically, Example \ref{ex:n} in Section~\ref{sec:lb} demonstrates that the $\log n$ dependence on $n$ is tight, and Example \ref{ex:cor-m} demonstrates that at least $\Omega(m)$ samples are required in order to form an approximate correlated equilibrium.

The arguments for proving Theorem \ref{theo:ce} are more direct than the Nash-equilibrium case  (i.e., Theorem \ref{theo:ne}) . A proof of Theorem~\ref{theo:ce} appears in Section \ref{sec:ce-proofs}.

\subsection{Existence of Simple Approximate Correlated Equilibrium}

We should emphasize again that the support of a correlated equilibrium is the number of \emph{action profiles} in the support of the equilibrium. Also, note that in an $n$-player $m$-action game the number of action profiles is $m^n$. If we use the existence of small-support approximate Nash equilibrium (which is also an approximate correlated equilibrium) we obtain the existence of an approximate correlated equilibrium with support of size $O((\log m + \log n)^n)$. 

By observing that Theorem \ref{theo:ce} holds with positive probability we can deduce the existence of approximate correlated equilibrium with support of size $O(m\log m + \log n)$. But can we obtain an approximate correlated equilibrium support size poly-logarithmic in $m$, instead of a polynomial? Example \ref{ex:cor-m} demonstrates that if we sample from an \emph{arbitrary} correlated equilibrium then we cannot. But, if we sample from a \emph{specific} approximate correlated equilibrium, then poly-logarithmic number of samples will be sufficient. It turns out that the specific approximate correlated equilibrium from which we should sample is an equilibrium in which \emph{each player} uses only a small number of her own actions in the support of the equilibrium. Existence of such an approximate correlated equilibrium is proved in Corollary \ref{cor:ne-ss} (because every approximate Nash equilibrium is also an approximate correlated equilibrium).

The following theorem shows that there always exists an approximate correlated equilibrium with support size poly-logarithmic in $n$ and $m$; moreover, the probabilistic structure of the equilibrium is simple: it is a $k$-uniform distribution (see Definition \ref{def:k-uni}).

\begin{theorem}\label{theo:ce-ss}
Every $n$-player $m$-action game admits a $k$-uniform $\varepsilon$-correlated equilibrium for every 
\begin{equation}
k> \frac{264}{\varepsilon^4}\ln m(\ln m + \ln n - \ln \varepsilon + \ln 16) = O(\log m (\log m + \log n)) 
\end{equation} 
\end{theorem}

A proof of this theorem, which uses the above mentioned ideas, appears in Section \ref{sec:ce-proofs}.

\subsection{Testing Approximate Correlated Equilibrium Play}

As in Section \ref{sec:ne-test}, we would like to design a test that uses $k$ samples to determine whether players are playing according to a $\delta$-correlated equilibrium or according to a distribution that is not a $(\delta + \varepsilon)$-correlated equilibrium.\footnote{We note again, that it is impossible to design such a test for distinguishing between $\delta$-correlated equilibrium and not a $\delta$-correlated equilibrium. This follows via arguments similar to the ones in Section \ref{sec:ne-test}.}

\begin{definition}
Given a number of samples $k \in \mathbb{Z}_+$, an \emph{$\varepsilon$-test with error probability $\alpha$ for $\delta$-correlated equilibrium that uses $k$-samples}, is a function $T:A^k\rightarrow \{\y, \n\}$, such that for every distribution $x \in \Delta(A)$ we have
\begin{itemize}
\item $\P(T((a(t))_{t\in [k]})= \y)\geq 1-\alpha$ for every $x$ that is a $\delta$-correlated equilibrium,

\item $\P(T((a(t))_{t\in [k]})= \n)\geq 1-\alpha$ for every $x$ that is not a $(\delta + \varepsilon)$-correlated equilibrium.
\end{itemize}
\end{definition}

The following theorem states that using $O(m\log m + \log n)$ samples we can test whether players are playing according to an approximate correlated equilibrium. Moreover, the test is quite natural, we need to simply check whether the empirical distribution of play is an approximate correlated equilibrium or not.

\begin{theorem}\label{theo:ce-test}
Let 
\begin{align*}
T((a(t)_{t\in [k]})=\begin{cases} 
\y &\text{if } s^k \text{ is a } (\delta + \frac{\varepsilon}{2}) \text{-correlated equilibrium} \\
\n &\text{otherwise}. 
\end{cases}
\end{align*}
Then $T$ is an $\varepsilon$-test with an $\alpha$-error-probability for $\delta$-correlated equilibrium, when the number of samples 
\begin{align*}
k & > \frac{8}{\varepsilon^2}(m\ln m + \ln n - \ln \alpha)=O(m\log m + \log n).
\end{align*}

\end{theorem}

Section \ref{sec:ce-proofs} contains a proof of this theorem. 

An unsatisfactory property of the above test is the polynomial dependence on the number of actions. Example \ref{ex:cor-m} in Section~\ref{sec:lb} demonstrate that the natural test that is presented in the theorem cannot use less then $\Omega(m)$ samples. Hypothetically, it could be the case that there exists some other test that uses significantly fewer samples. The following theorem states that this is not the case. The number of samples must be polynomial in $m$, even for the case where $\varepsilon$ and $\alpha$ are constants.

\begin{theorem}\label{theo:cor-im}
Every $\frac{1}{2}$-test with an error probability $\frac{1}{4}$ for exact correlated equilibrium for two-player $m$-action games must use at least $\sqrt{\frac{m}{2}}$ samples.
\end{theorem}

See Section \ref{sec:ce-proofs} for a proof. 

\subsection{Proofs}\label{sec:ce-proofs}

\begin{proof}[Proof of Theorem \ref{theo:ce}]

Note that $R^i_f(a)$ where $a\sim x$ is a random variable that assumes values in $[-1,1]$, and $\E_{a\sim s^k} [R^i_f(a)]=\frac{1}{k} R^i_f(a(t))$ is the average regret on the samples. Since $x$ is a correlated equilibrium we know that $\E_{a\sim x} [R^i_f(a)]\leq 0$. Therefore by Hoeffding's inequality and the choice of $k$ we have
\begin{equation*}
\P(\E_{a\sim s^k} [R^i_f(a)] \geq \varepsilon)\leq e^{-\frac{\varepsilon^2}{2}k}\leq \frac{\alpha}{nm^m}.
\end{equation*}
For every player $i$, there are $m^m$ switching rules of the form $f:A_i \rightarrow A_i$. Hence, summing across $n$ players, we get that the total number of different switching rules is $nm^m$. Therefore, the union bound implies that with probability greater than  $1-\alpha$ we have $\E_{a\sim s^k} [R^i_f(a)] < \varepsilon $. Hence, with probability at least $1-\alpha$,  the empirical distribution $s^k$ is an $\varepsilon$-correlated equilibrium.

\end{proof}

\begin{proof}[Proof of Theorem \ref{theo:ce-ss}]

By Corollary \ref{cor:ne-ss}, there exists an $\frac{\varepsilon}{2}$-Nash equilibrium $x$ where every player $i$ uses at most  $b=\left\lceil \frac{32}{\varepsilon^2}(\ln n + \ln m - \ln \varepsilon + \ln 16) \right\rceil$ actions from $A_i$. We denote the set of player's $i$ actions that are played with positive probability in $x$ by $B_i$, where $|B_i|\leq b$. Let us implement the sampling method for the distribution $x$ which is an $\frac{\varepsilon}{2}$-Nash equilibrium, and therefore, also an $\frac{\varepsilon}{2}$-correlated equilibrium.

Since $\E_{a\sim x} [R^i_f(a)] \leq \frac{\varepsilon}{2}$ by Hoeffding's inequality, we have
\begin{equation}\label{eq:incor}
\Pr (\E_{a\sim s^k} [ R^i_f(a) ] \geq \varepsilon ) \leq e^{-\frac{\varepsilon^2}{8}k}.
\end{equation}

Note that $s^k$ is an $\varepsilon$-correlated equilibrium iff $\E_{a\sim s^k} [R^i_f(a)] \leq \varepsilon$ for every switching rule $f:B_i\rightarrow A_i$ (note that the number of such switching rules is at most $m^b$ for every player). In other words, we can consider only switching rules $f:B_i\rightarrow A_i$ instead of $f:A_i\rightarrow A_i$, because all the recommendations to player $i$ will be from the set $B_i$. 

The choice of $k$ guarantees that
\begin{equation}
e^{-\frac{\varepsilon^2}{8}k}<\frac{1}{nm^b}.
\end{equation}
Therefore, using inequality (\ref{eq:incor}) and the union bound, we get that with positive probability $\E_{a\sim s^k}  [R^i_f(a)]  \leq \varepsilon$ is satisfied for every $f:B_i \rightarrow A_i$, which implies that $s^k$ is an $\varepsilon$-correlated equilibrium. This implies, via the probabilistic method, that such a $k$-uniform correlated equilibrium exists.
\end{proof}

\begin{proof}[Proof of Theorem \ref{theo:ce-test}]
Hoeffding's inequality and the choice of $k$ guarantee that 
\begin{equation*}
\P(|\E_{a\sim x} [R^i_f(a)] - \E_{a\sim s^k} [R^i_f(a)]|\geq \frac{\varepsilon}{2})\leq 2e^{-\frac{\varepsilon}{8}k}<\frac{\alpha}{nm^m}.
\end{equation*}
The total number of switching rules is $nm^m$, therefore with probability of at least $1-\alpha$ we have
\begin{equation}\label{eq:cor-in}
|\E_{a\sim x} [ R^i_f(a) ] - \E_{a\sim s^k} [R^i_f(a)]|< \frac{\varepsilon}{2},
\end{equation}
for all players $i$ and all switching rules $f:A_i\rightarrow A_i$. Throughout the reset of the paper, we assume that inequality (\ref{eq:cor-in}) is satisfied for every switching rule.

If $x$ is a $\delta$-correlated equilibrium then
\begin{equation*}
\E_{a\sim s^k} [R^i_f(a)]\leq \E_{a\sim x} [R^i_f(a)]+ \frac{\varepsilon}{2} \leq \delta +\frac{\varepsilon}{2},
\end{equation*}
which means that $s^k$ is an $(\delta + \frac{\varepsilon}{2})$-correlated equilibrium.

If $x$ is not a $(\delta+\varepsilon)$-correlated equilibrium then there exists a player $i$ and a switching rule $f^*$ such that $\E_{a\sim x} [R^i_{f^*}(a) ] > \delta+\varepsilon$. So,
\begin{equation*}
\E_{a\sim s^k} [R^i_{f^*}(a)]\geq \E_{a\sim x} [R^i_{f^*}(a)]- \frac{\varepsilon}{2} > \delta + \frac{\varepsilon}{2},
\end{equation*}
which means that $s^k$ is not a $(\delta + \frac{\varepsilon}{2})$-correlated equilibrium.

\end{proof}

\begin{proof}[Proof of Theorem \ref{theo:cor-im}]
Instead of proving that in $2$-players $m$-actions games every test must use at least $k=\sqrt{\frac{m}{2}}$ samples, we prove the equivalent statement that in $2$-players $(2m)$-actions games every test must use at least $k=\sqrt{m}$ samples. 

We consider the same game as in Example \ref{ex:cor-m} in Section~\ref{sec:lb}: consider a two-players $2m$-actions zero-sum game in which the players are playing matching-pennies, but in addition to player's $i$ ``real'' action $r_1\in [2]$ player $i$ also chooses a ``dummy'' action $d_i\in [m]$ which does not influence the payoff. Formally, the payoff functions of the players are defined by
\begin{align*}
u_1((r_1,d_1),(r_2,d_2))=1-u_2((r_1,d_1),(r_2,d_2))=\begin{cases}
1 &\text{ if } r_1=r_2,\\
0 &\text{ otherwise.}
\end{cases}
\end{align*}

Consider the correlated equilibrium $x$ where $x((r_1,d),(r_2,d))=\frac{1}{4m}$, for every $d\in [m]$ and every $r_1,r_2\in [2]$. In other words, $x$ is the correlated equilibrium where beyond the actual $(\frac{1}{2},\frac{1}{2})$ play of the real matching-pennies, the players always choose the same dummy action.


Let $b:=(b_d)_{d\in [m]}$ be a vector of size $m$, where each coordinate $b_d$ is a pair $b_d\in \{(1,1),(1,2),(2,1),(2,2)\}$. We define the distributions $y_b$ (we have $4^m$ such distributions) by:
\begin{equation}
y_b((r_1,d_1),(r_2,d_2))=\begin{cases}
\frac{1}{m} &\text{ if } d_1=d_2=d \text{ and } (r_1,r_2)=b_d,\\
0 &\text{ otherwise.}
\end{cases}
\end{equation}
Loosely speaking, the distribution $y_b$ picks for every $d\in [m]$ a single action $(r_i,r_j)$ for both players and puts a measure of $\frac{1}{m}$ on it. This is in contrast to $x$ which puts an equal measure of $\frac{1}{4m}$ on all four actions $(r_i,r_j)$.

Let $\omega$ be the event $\omega:=\{((r(t),d(t))_{t\in [k]}: d(t)\neq
d(t') \text{ for } t\neq t'\}$; i.e., all the samples have different values of $d$. Note that $\P_x(\omega)=\P_{y_b}(\omega)$ for every $b$, because the event $\omega$ depends only on the samples of $d$, and both $x$ and $y_b$ have the uniform distribution over the values of $d$.

We claim that if $k=\lfloor\sqrt{m} \rfloor$ then $\P_x(\omega)=\P_{y_b}(\omega)>\frac{1}{2}$ (for every $b$).

The $t$th sample will have the same value $d$ as one of the previous with probability of at most $\frac{t-1}{m}$. Using the union bound we get that
\begin{equation*}
1-\P_x(\omega)\leq \frac{0}{m}+\frac{1}{m}+\frac{2}{m}+...+\frac{\lfloor\sqrt{m}\rfloor-1}{m}\leq \frac{(\sqrt{m}-1)\sqrt{m}}{2m}<\frac{1}{2}.
\end{equation*}

A test with error-probability $\frac{1}{4}$ should return with probability $\frac{3}{4}$ the answer $\y$ for the correlated equilibrium $x$, and it should return the answer $\n$ with probability $\frac{3}{4}$ for all the distributions $y_b$ which are not $\frac{1}{2}$-correlated equilibria. In particular, if we first draw the distribution from which we sample (according to some probability distribution), and then sample from the chosen distribution, the probability of an error of the test should be less than $\frac{1}{4}$ (because for each one of the distributions that we draw the probability of error is less than $\frac{1}{4}$).
Let us draw the distribution from which we sample as follows. The distribution $x$ is chosen with probability $\frac{1}{2}$, and each one of the distributions $y_b$ is chosen with probability $\frac{1}{2\cdot4^m}$. If the sequence of samples is $(r(t),d(t))_{t\in [k]}\in \omega$ then, by the symmetry of the distributions $\{y_b\}_b$, the probability that it is sampled from $x$ is equal to the probability that it is sampled from one of the distributions $y_b$. Therefore, for sequences of samples in $\omega$ the test makes an error with probability of at least $\frac{1}{2}$, and sequence of samples is in $\omega$ with probability of at least $\frac{1}{2}$. Therefore, the probability of an error is at least $\frac{1}{4}$.

\end{proof}

\section{Approximate Coarse Correlated Equilibrium}

The case of coarse correlated equilibrium is the simplest one. Here we present the results for coarse correlated equilibria without the proofs, since they are quite similar to the proofs of the correlated equilibrium results presented in Section \ref{sec:ce}.

The key difference between coarse correlated equilibrium and correlated equilibrium is that in coarse correlated equilibrium every player is allowed to deviate to one fixed pure action (irrespective of the recommendations of the mediator), instead of allowing the player to deviate to different actions for different recommendations.

\begin{definition}\label{defn:cce}
For every pure action $j\in A_i$ we denote by $R^i_j(a):=u_i(j,a_{-i})-u_i(a_i,a_{-i})$ the regret of player $i$ for not choosing the action $j$ at strategy profile $a$. 

A distribution $x\in \Delta(A)$ is an $\varepsilon$-\emph{coarse correlated equilibrium} if $\E_{a\sim x} [ R^i_j(a) ] \leq \varepsilon$ for every player $i$ and every action $j\in A_i$.
\end{definition}

\begin{theorem}\label{theo:cce}
Let $x$ be a coarse correlated equilibrium of an $n$-player $m$-action game and parameters $\varepsilon, \alpha \in (0,1)$. Then, the empirical distribution $s^k$ defined over $k$ i.i.d.~samples from $x$ is an $\varepsilon$-coarse correlated equilibrium with probability greater than $1-\alpha$ for every
\begin{equation*}
k> \frac{2}{\varepsilon^2}(\ln m + \ln n - \ln \alpha)=O(\log m + \log n).
\end{equation*}
\end{theorem}

We can establish this theorem using the same ideas as in the proof of Theorem \ref{theo:ce}. The only difference is that here, instead of $nm^m$ inequalities (one for every $R^i_f$, where $f:A_i \rightarrow A_i$),  we need to satisfy only $nm$ inequalities, one for every $R^i_j$.

This theorem gives us the following result that establishes the existence of simple approximate coarse correlated equilibrium:

\begin{corollary}\label{cor:cce-ss}

Every $n$-player $m$-action game admits a $k$-uniform $\varepsilon$-coarse correlated equilibrium for every 
\begin{equation}
k> \frac{2}{\varepsilon^2}(\ln m + \ln n)=O(\log m + \log n). 
\end{equation} 
\end{corollary}
This bound is tight, see Examples \ref{ex:n} and \ref{ex:al}.

Analogous to Theorem \ref{theo:ce-test}, we have the following result for testing approximate coarse correlated equilibrium play.

\begin{corollary}\label{cor:cce-test}
Let 
\begin{equation*}
T((a(t)_{t\in [k]})=\begin{cases} 
\y &\text{if } s^k \text{ is a } (\delta + \frac{\varepsilon}{2}) \text{-coarse correlated equilibrium}, \\
\n &\text{otherwise}. 
\end{cases}
\end{equation*}
Then $T$ is an $\varepsilon$-test with error probability $\alpha$ for $\delta$-coarse correlated equilibrium for
\begin{equation*}
k> \frac{8}{\varepsilon^2}(\ln m + \ln n - \ln \alpha)=O(\log m + \log n).
\end{equation*}

\end{corollary}

The theorem follows from a proof similar to the one for Theorem \ref{theo:ce-test}.

\section{Lower Bounds}\label{sec:lb}

In this section we present lower bounds for the number of samples from an equilibrium that are required in order that the empirical distribution of play will be an approximate equilibrium (with high probability). 

The following example demonstrates that $\Omega(\log n)$ samples are required for all cases: Nash equilibrium, correlated equilibrium, and coarse-correlated equilibrium.

\begin{example}\label{ex:n}
Consider the following $2n$-players two-actions game. We have $n$ pairs of players $(p^1_i,p^2_i)_{i\in [n]}$. Player $p^j_i$ is playing matching-pennies with his partner $p^{3-j}_i$ (the actions of the pair $(p^1_i,p^2_i)$ have no influence on the payoffs of other pairs). Consider the Nash equilibrium where each player is playing $(\frac{1}{2},\frac{1}{2})$ (which is also a correlated equilibrium and a coarse-correlated equilibrium). If the number of samples is $k\leq \frac{\log n}{2}$ then the probability that player's $p^j_i$ empirical distribution of play will be a pure strategy is
\begin{equation*}
2\left(\frac{1}{2}\right)^{\frac{\log n}{2}}\geq \frac{1}{\sqrt{n}}.
\end{equation*}
Therefore the probability that no player will have a pure-strategy empirical distribution is at most
\begin{equation*}
\left(1-\frac{1}{\sqrt{n}}\right)^{2n}\approx e^{-2\sqrt{n}}.
\end{equation*}
Note that the requirement that no player will have a pure-strategy empirical distribution is necessary for the empirical distribution of play to be a $\frac{1}{2}$-coarse correlated equilibrium (and therefore it is necessary also for $\frac{1}{2}$-correlated and $\frac{1}{2}$-Nash equilibria). It follows that the probability that the empirical distribution of play will form a $\frac{1}{2}$-equilibrium is exponentially small in $n$.

\end{example}

The following example of Alth\"{o}fer \cite{A} demonstrates that $\Omega(\log m)$ samples are required for all cases: Nash equilibrium, correlated equilibrium and coarse-correlated equilibrium (actually for the correlated equilibrium case Example \ref{ex:cor-m} will demonstrate a much stronger result).

\begin{example}\label{ex:al}
Let $m=4^b$ for $b\in \mathbb{N}$, and consider the following two-players $m$-actions zero-sum game.

Player 1 picks an element $i\in [2b]$ (player 1 has $2b<m$ actions).

Player 2 picks a subset of $S_j\subset [2b]$ such that $|S_j|=b$ (player 2 has $\binom{2b}{b}<m$ actions).

The payoffs are defined by
\begin{equation*}
u_2(i,S_j)=-u_1(i,S_j)=\begin{cases}
1 &\text{ if } i\in S_j,\\
0 &\text{ otherwise.}
\end{cases}
\end{equation*}

Player $1$ can guarantee to pay at most $\frac{1}{2}$ by playing the uniform distribution. If in the support of the distribution (which might be correlated) player $1$ plays at most $b$ different actions, then player $2$ has a pure strategy that will yield a payoff of $1$; therefore in every $\frac{1}{4}$-equilibrium (Nash, correlated or coarse-correlated) player $1$ should play at least $b+1$ different strategies. Therefore, in order that the empirical distribution  will be a $\frac{1}{4}$-equilibrium the number of samples must be greater than $b=\frac{\log m}{2}$. 

\end{example}

The following example demonstrates that $\Omega(m)$ samples are required for the case of correlated equilibrium.

\begin{example}\label{ex:cor-m}
Consider the following two-players $2m$-actions zero-sum game. The players are playing matching-pennies, but in addition to player's $i$ ``real'' action $r_1\in [2]$ player $i$ also chooses a ``dummy'' action $d_i\in [m]$ which does not influence the payoff. Formally, the payoff functions of the players are defined by
\begin{equation*}
u_1((r_1,d_1),(r_2,d_2))=1-u_2((r_1,d_1),(r_2,d_2))=\begin{cases}
1 &\text{ if } r_1=r_2,\\
0 &\text{ otherwise.}
\end{cases}
\end{equation*}

Consider the correlated equilibrium $x$ where $x((r_1,d),(r_2,d))=\frac{1}{4m}$ for every $d\in [m]$ and every $r_1,r_2\in [2]$. In other words, $x$ is the correlated equilibrium where beyond the actual $(\frac{1}{2},\frac{1}{2})$ play of the real matching-pennies, the players always choose the same dummy action.

If the number of samples is $k=m$, then for any $d \in [m]$ the probability that it is picked exactly once during the sampling is $m \frac{1}{m} \cdot \left(1-\frac{1}{m} \right)^{m-1}\approx \frac{1}{e}$. If a certain $d$ was picked exactly once then both players can deduce from $d$ which action their opponent will play. Note that the expected number of $d\in [m]$ that are sampled exactly once is $\frac{m}{e}$. Moreover, the probability that the number of exactly-once-sampled $d$'s will be smaller than $\frac{m}{2e}$ is exponentially small in $m$ (see , e.g., \cite{FM}, Lemma 4). So, with probability that is exponentially close to 1, in the resulting uniform distribution at least one player may increase it's payoff by at least $\frac{1}{4e}$ by reacting optimally to the opponent's known strategy in all cases where she got the recommendation $(r_i,d)$ where $d$ was chosen exactly once. Therefore the empirical distribution of samples is an $\frac{1}{4e} $-correlated equilibrium with probability exponentially small in $m$.  

\end{example}

The focus of the paper was on the dependence of the number of samples on $m$ and $n$. However, Theorems~\ref{theo:ne},~\ref{theo:ce}, and ~\ref{theo:cce} proves also a dependence on $\varepsilon$. For the case of Nash equilibrium, Theorem~\ref{theo:ne} proves a bound of $O(\frac{1}{\varepsilon^2} \log (\frac{1}{\varepsilon}))$. Theorems~\ref{theo:ce} and~\ref{theo:cce} proves a bound of $O(\frac{1}{\varepsilon^2})$ for correlated and coarse-correlated equilibrium. The following example demonstrates that those bounds are tight (except for the case of Nash equilibrium where is it almost tight).

\begin{example}
Consider the matching-pennies game, with the unique Nash equilibrium $((\frac{1}{2},\frac{1}{2})(\frac{1}{2},\frac{1}{2}))$. A necessary condition for the empirical distribution of play to form an $\varepsilon$-equilibrium (Nash correlated or coarse-correlated) is that the empirical distribution of player 1 should be $(p,1-p)$ where $p\in [\frac{1}{2}-\varepsilon,\frac{1}{2}+\varepsilon]$. By the central limit theorem, after $k$ samples, with constant probability the deviation from the expectation ($p=\frac{1}{2}$) is at least $\frac{1}{\sqrt{k}}$. Therefore, if we draw $k$ samples for $k<\frac{1}{\varepsilon^2}$, then with positive probability the deviation from $\frac{1}{2}$ will be at least $\frac{1}{\sqrt{k}}>\varepsilon$.
\end{example}

\section{Discussion}\label{sec:disc}

\subsection{Sampling from One Type of Equilibrium to Achieve Another} 
In this paper we considered three types of equilibria: Nash, correlated, and coarse correlated. Our high level approach was to sample from an equilibrium of a particular type to generate an approximate equilibrium of the same type. We can modify this approach a bit and, in principle, ask the following question: How many samples from an equilibrium of a particular type are required to generate an approximate equilibrium of a different type?  


Note that the notion of coarse correlated equilibrium is a generalization of correlated equilibrium, and the latter generalizes Nash equilibrium. In general, we cannot hope to get a more refined notion of equilibrium by sampling from a more general one. But, hypothetically, it might be the case that fewer samples from a refined equilibrium type are sufficient for generating an approximate equilibrium of a more general type.

First we observe that $\Omega(\log n + \log m)$ samples are necessarily required to generate a coarse correlated equilibrium, even if the samples are drawn from a Nash equilibrium. This follows from the lower bound of Example \ref{ex:n} (wherein we actually sample from a Nash equilibrium) and Example \ref{ex:al} (in which the counting argument holds \emph{irrespective} of the initial distribution). 

So the remaining question is, can $o(m)$ samples from a Nash equilibrium generate an approximate correlated equilibrium? In other words, can we overcome the $\Omega(m)$ sampling lower bound established in Example \ref{ex:cor-m}? The answer to this question is \emph{no}. In particular, consider the same game as in Example \ref{ex:cor-m}, but now draw $m$ samples from the Nash equilibrium where both plays are playing the uniform distribution over their $2m$ actions. We say that a pair $(d_1,d_2)$ \emph{appears exactly once} if the pair $(d_1,d_2)$ appears in the sample (i.e., one of the samples is $((r_1,d_1),(r_2,d_2))$ for some $r_1,r_2\in [2]$) and, for $i=1,2$, $d_i$ appears exactly once among all the samples.

For every pair $(d_1,d_2)$, the probability that it appears exactly once is equal to $m\frac{1}{m^2}(1-\frac{2m-1}{m^2})^m \approx \frac{1}{e^2 m}$. Therefore, the expected number of recommendation pairs, $(d_1,d_2)$, that appear exactly once is $\frac{m}{e^2}$. Hence, among the $m$ samples a significant fraction of pairs appear exactly once. Note that if a recommendation appears exactly once, then both of the players can deduce their opponent's strategy from the recommendation, which cannot occur in an equilibrium.

\subsection{Hypothesis Testing}

The present paper proposes simple intuitive tests to whether players are following some (approximate) equilibrium or behaving far from any equilibrium. The notion of distance that we use, $\varepsilon$-equilibrium, is given in terms of the game incentives. This is in contrast to the general-purpose hypothesis-testing literature that refers to distances between distributions, such as the total variation norm.

Putting our problem in a more general context, gives rise to a question of independent interest.
\begin{question}\label{q-polytope}
 Consider the set of all probability measures on $d$ elements, $\Delta(d)$. Given a polytope $P\subset\Delta(d)$ and $\varepsilon>0$, how many independent samples from an unknown distribution $q\in\Delta(d)$ are needed to ascertain with high probability that either
\begin{itemize}
\item $H_1$: $q$ is in $P$, or 
\item $H_2$: the total variation distance $\mathrm{dist}(q,P)>\epsilon$ ?
\end{itemize}
\end{question}

Upper bounds for Question~\ref{q-polytope} translate automatically to upper bounds for our correlated and coarse-correlated equilibrium testing; this is because the game payoff function is $1$-Lipschitz in the total-variation norm. 

It should be noted that the upper bounds obtained in this paper are stronger than what one could hope to deduce from total variation estimates. A special case of Question~\ref{q-polytope}, when $P$ consists of a single point, is known as the ``multinomial goodness of fit'' problem in hypothesis-testing literature \cite{siotani1984asymptotic, cressie1984multinomial}, and as the ``identity testing'' problem in the computer science literature \cite{goldreich2011testing,batu2000testing,batu2013testing,batu2001testing,chan2014optimal,valiant2013instance}. For this special case the answer to Question~\ref{q-polytope} is $\Theta(\sqrt d)$. Note that this special case corresponds to exact equilibrium testing ($\delta=0$) in games with a unique equilibrium. Here the dimension is $d=m^n$; therefore an upper bound of $m^{\frac{n}{2}}$ is obtained. Even in this particular case, the total variation bound is generally  weaker than what we obtained. The only exception is the case of correlated equilibrium in two-player games, where our upper bound is $O(m\log m)$, and the total variation result suggest the possibility of an $O(m)$ upper bound.  

Lower bounds for Question~\ref{q-polytope} do not directly translate to equilibrium testing, because being close in payoffs does not mean being close in total variation distance. For example, a full support equilibrium can be approximated  (in terms of payoffs) by mixed strategies of much smaller support.

\subsection{Future Work}

This paper establishes tight bounds on the rate of convergence of the empirical distribution (of equilibrium play) to an approximate equilibrium. These bounds imply the existence of  \emph{small-support} approximate equilibria. But, whether our poly-logarithmic upper bounds on support size are tight  remains an open question. Note that the  $\log m$ lower bound developed in Example \ref{ex:al} applies to support size as well. In particular, Example \ref{ex:al} establishes that there does not exist an $\varepsilon$-equilibrium (Nash, correlated, or coarse correlated) with support size smaller than $\log m$. However, to the best of our knowledge, lower bounds on support size of approximate equilibrium (Nash, correlated, or coarse correlated) that depends on $n$ have not been established.

\begin{question} Let $k=k(n,m,\varepsilon)$ be the smallest number such that every $n$-player $m$-action game admits a $k$-uniform $\epsilon$-equilibrium. Fix $\epsilon>0$ and $m\in \mathbb{N}$ (i.e., we refer to them as constants), and denote $f(n)=f(n,m,\varepsilon)$. What is the asymptotic behavior of $k(n)$? In particular, does $\lim_{n \rightarrow \infty} k(n)=\infty$? The question remains open for all three equilibrium types: Nash, correlated, and coarse correlated. 
\end{question}

\section*{Acknowledgements}
The authors would like to thank Tu\u{g}kan Batu and Federico Echenique for their constructive suggestions and comments.  

\bibliographystyle{plain}
\bibliography{sampling-testing}

\appendix
\section{Concentration Inequality for Product Distributions}

Here we generalize the classic Hoeffding's inequality \cite{H} to product probability spaces. In order to state our generalization we need the following definition.

\begin{definition}
Let $(\Omega,\mu)$ be a discrete probability space. The \emph{$k$-sample approximation of $\mu$} is the random $k$-uniform distribution $\mu^{(k)}\in\Delta_k(\Omega)$ given by taking the average of $k$ i.i.d. samples from $\mu$.

Formally, one can implement $\mu^{(k)}$ by taking $k$ independent random variables $x_1,\ldots,x_k$ assuming values in $\Omega$ with distribution $\mu$. By identifying $\Omega$ with $\Delta_1(\Omega)$, $\mu^{(k)}$ is given as the $\Delta_k(\Omega)$-valued random variable
\[
\mu^{(k)}=\frac 1 k\sum_{i=1}^kx_i.
\] 
\end{definition}

The classic Hoeffding's Inequality can be stated as follows:

\begin{theorem}[Hoeffding 1963]\label{theo:hoe}
Let $(\Omega,\mu)$ be a (discrete) probability space. For every $\varepsilon >0$, $k\in\mathbb N$, and $f:\Omega\to[0,1]$
\[
\P(|\E_{\mu^{(k)}}[f]-\E_\mu[f]|>\varepsilon)\leq 2e^{-\frac{\varepsilon^2}{2}k}.
\]
\end{theorem}

When $(\Omega,\mu)$ is the product of $n$ probability spaces, we obtain a similar inequality in which the approximation is achieved by a product measure and the error term is independent of $n$.
  
\begin{theorem}\label{theo:hoe-gen}
Let $(\Omega_1,\mu_1),\ldots,(\Omega_n,\mu_n)$ be discrete probability spaces. Consider the product space $(\Omega=\prod_i\Omega_i,\mu=\prod_i\mu_i)$. For every $\varepsilon >0$, $k\in\mathbb N$, and $f:\Omega\to[0,1]$
\[
\P(|\E_{\prod_i\mu_i^{(k)}}[f]-\E_\mu[f]|>\varepsilon)\leq \frac{4e^{-\frac{\varepsilon^2}8k}}{\varepsilon}.
\]
\end{theorem}

Lemma~\ref{pro:main} follows immediately from Theorem~\ref{theo:hoe-gen}.
\begin{proof}[Proof of Lemma~\ref{pro:main}]
The strategy profile $s_{-i}^k$ is just the $k$-sampling approximation of $x_{-i}$.
\end{proof}

\begin{proof}[Proof of Theorem~\ref{theo:hoe-gen}]
For every $i\in[n]$, let us implement each $\mu_i^{(k)}$ as the average of the $k$ i.i.d. random variables $x_1^i,\ldots,x_k^i$ prescribed by the definition of $k$-sampling approximation. We begin by rewriting $\E_\mu[f]$. For every $l\in[k]$, we can write
\begin{equation*}
\E_\mu[f]=\frac{1}{k^{n}}\underset{j_1,j_2,...,j_n \in [k]}{\sum} \E[f(x_{j_1+l}^1,...,x_{j_n+l}^n)],
\end{equation*}
where the indexes $j_i+l$ are taken modulo $k$. If we take the average over all possible $l$ we have
\begin{equation}\label{eq:pay}
\E_\mu[f]=\frac{1}{k^{n}}\underset{j_1,j_2,...,j_n \in [k]}{\sum} \frac{1}{k} \underset{l\in[k]}{\sum} \E[f(x_{j_1+l}^1,...,x_{j_n+l}^n)]. 
\end{equation}
For every initial multi-index $j_*=(j_1,j_2,...,j_n)\in [k]^{n}$ and every $l\in [k]$, we denote $x_{j_*+l}:=(x_{j_1+l}^1,...,x_{j_n+l}^n)\in \Omega$, and we define the random variable 
\begin{equation}\label{eq:d}
d(j_*):= \begin{cases}
0 & \text{if } \left\lvert \frac{1}{k}\underset{l\in[k]}{\sum}\E[f(x_{j_*+l})]-\E_\mu[f] \right\rvert \leq \dfrac{\varepsilon}{2}, \\
1 & \text{otherwise.}
\end{cases}
\end{equation}
By the definition of $d(j_*)$, we have
\begin{equation}\label{eq:din}
d(j_*)+\frac{\varepsilon}{2} \geq \left\lvert \frac{1}{k}\underset{l\in[k]}{\sum}\E[f(x_{j_*+l})]-\E_\mu[f] \right\rvert.
\end{equation}
Note also that for any fixed $j_*$, the random variables $x_{j_*+1},x_{j_*+2}, \ldots,x_{j_*+k}$ are independent with distribution $\mu$; therefore their average implements $\mu^{(k)}$; therefore, by Hoeffding's inequality, we have
\begin{equation}\label{eq:hof}
\E[d(j_*)]\leq 2 e^{-\frac{\varepsilon^2}{8}k}.
\end{equation}
Using representation (\ref{eq:pay}) of $\E_\mu[f]$ and inequalities (\ref{eq:din}) and (\ref{eq:hof}), we get
\begin{equation}\label{eq:fin}
\begin{aligned}
&\P(|\E_{\prod_i(\mu_i^{(k)})}[f]-\E_\mu[f]|\geq \varepsilon)= \\ 
&= \P \left( \left\lvert \frac{1}{k^{n}} \underset{j_* \in [k]^{n}}{\sum} \frac{1}{k} \underset{l\in[k]}{\sum} \E[f(x_{j_*+l})]-\E_\mu[f] \right\rvert \geq \varepsilon \right) \\
&\leq \P \left( \frac{1}{k^{n}} \underset{j_* \in [k]^{n}}{\sum} \left\lvert \frac{1}{k} \underset{l\in[k]}{\sum} \E[f(x_{j_*+l})]-\E_\mu[f] \right\rvert \geq \varepsilon \right) \\
&\leq \P \left( \frac{1}{k^{n}} \underset{j_* \in [k]^{n}}{\sum} d(j_*) \geq \frac{\varepsilon}{2} \right) \leq \frac{4 e^{-\frac{\varepsilon^2}{8}k}}{\varepsilon},
\end{aligned}
\end{equation}
where the last inequality follows from Markov's inequality.
\end{proof}

\end{document}